\documentclass[aip,jmp,amsmath,amssymb,reprint,nofootinbib]{revtex4-2}

\usepackage[utf8]{inputenc}
\usepackage[unicode=true,bookmarksopen=true,colorlinks=true,urlcolor=blue,anchorcolor=blue,citecolor=blue,filecolor=blue,linkcolor=blue,menucolor=blue,pagecolor=blue,linktocpage=true,pdfa=true]{hyperref}

\usepackage{mathtools}
\usepackage{amsthm}
\usepackage{lmodern}
\usepackage[T1]{fontenc}
\usepackage{amsfonts}
\usepackage{mathrsfs}
\usepackage{bbm}
\DeclareMathAlphabet{\mathbfi}{OML}{cmm}{b}{it}

\let\originalleft\left
\let\originalright\right
\renewcommand{\left}{\mathopen{}\mathclose\bgroup\originalleft}
\renewcommand{\right}{\aftergroup\egroup\originalright}
\makeatletter
\newcommand{\biggg}{\bBigg@\thr@@}
\newcommand{\Biggg}{\bBigg@{3.5}}
\makeatother

\makeatletter
\newenvironment{equations}[1][]{\subequations\ifx\relax#1\relax\else\label{#1}\fi\align\ignorespaces}{\endalign\ignorespacesafterend\endsubequations}
\def\@spliteq#1{\begin{equation}\begin{split}#1\end{split}\end{equation}}
\def\@spliteqstar#1{\begin{equation*}\begin{split}#1\end{split}\end{equation*}}
\def\splitequation{\collect@body\@spliteq}
\expandafter\def\csname splitequation*\endcsname{\collect@body\@spliteqstar}

\expandafter\def\csname endsplitequation*\endcsname{\ignorespacesafterend}
\makeatother

\renewcommand{\vec}[1]{{\ifnum9<1#1\mathbf{#1}\else\ifcat\noexpand#1\relax\boldsymbol{#1}\else\mathbfi{#1}\fi\fi}}
\newcommand{\mathe}{\mathrm{e}}
\newcommand{\mathi}{\mathrm{i}}
\let\oldre\Re
\let\oldim\Im
\renewcommand{\Re}{\oldre\mathfrak{e}\,}
\renewcommand{\Im}{\oldim\mathfrak{m}\,}
\newcommand{\total}{\mathop{}\!\mathrm{d}}
\newcommand{\laplace}{\mathop{}\!\bigtriangleup}
\newcommand{\abs}[1]{{\left\lvert{#1}\right\rvert}}
\newcommand{\norm}[1]{{\left\lVert{#1}\right\rVert}}
\newcommand{\sgn}{\operatorname{sgn}}

\newcommand{\1}{\mathbbm{1}}

\newcommand{\eqend}[1]{\,#1}
\newcommand{\bigo}[1]{\mathcal{O}\left({#1}\right)}

\newcommand{\bra}[1]{\left\langle{#1}\right\vert}
\newcommand{\ket}[1]{\left\vert{#1}\right\rangle}
\newcommand{\expect}[1]{\left\langle{#1}\right\rangle}
\newcommand{\bessel}[3]{\mathop{}\!\mathrm{#1}_{#2}\left(#3\right)}

\newtheorem{theorem}{\textsc{Theorem}}[section]

\newtheorem{proposition}[theorem]{\textsc{Proposition}}

\newtheorem{corollary}[theorem]{\textsc{}Corollary}
\newtheorem{definition}[theorem]{\textsc{Definition}}

\newtheorem{remark}[theorem]{Remark}

\newtheorem{assumption}[theorem]{Assumption}

\begin{document}

\title{Strict deformations of quantum field theory in de Sitter spacetime}

\author{M. B. Fr{\"o}b}
\email{mfroeb@itp.uni-leipzig.de}
\author{A. Much}
\email{much@itp.uni-leipzig.de}
\affiliation{Institut f{\"u}r Theoretische Physik, Universit{\"a}t Leipzig, Br{\"u}derstra{\ss}e 16, 04103 Leipzig, Germany}

\date{05. February 2021}
\revised[Revised: ]{14. May 2021}

\begin{abstract}
We propose a new deformed Rieffel product for functions in de Sitter spacetime, and study the resulting deformation of quantum field theory in de Sitter using warped convolutions. This deformation is obtained by embedding de Sitter in a higher-dimensional Minkowski spacetime, deforming there using the action of translations and subsequently projecting back to de Sitter. We determine the two-point function of a deformed free scalar quantum field, which differs from the undeformed one, in contrast to the results in deformed Minkowski spacetime where they coincide. Nevertheless, we show that in the limit where de Sitter spacetime becomes flat, we recover the well-known non-commutative Minkowski spacetime.
\end{abstract}


\maketitle

\section{Introduction}

The core statement of the geometrical measurement problem given in Refs.~\onlinecite{Ahluwalia1993,DFR} is that spacetime, locally modelled as a manifold, should break down at very short distances of the order of the Planck length. This a consequence of the limitations in possible accuracies of localization of spacetime points, due to the interplay of the quantum-mechanical uncertainty principle with the formation of black holes in general relativity. Hence, a true theory of quantum gravity has to take the loss of classical spacetime into account. A way out of this dilemma and a path to put these arguments into a mathematical framework, is the proposition of quantizing spacetime itself. Hence, a typical starting point is to represent spacetime by (representations of) an algebra $\mathcal{V}$ with generators $\hat{x}$ obeying the following commutation relations:
\begin{equation}
\label{eq:commutator_mink}
\left[ \hat{x}^\mu, \hat{x}^\nu \right] = 2 \mathi \Theta^{\mu\nu} \eqend{,}
\end{equation}
where $\Theta$ is a constant skew-symmetric matrix and $\mu,\nu = 0,1,\ldots,n$. In this way, the commutative smooth manifold structure becomes non-commutative, and one can derive an uncertainty principle for the measurement of coordinates.

In a next step, one has to define a quantum field theory (QFT) in the non-commutative spacetime described by $\mathcal{V}$, which was also done in a mathematically rigorous fashion in Ref.~\onlinecite{DFR}. The representation space of a QFT in the given non-commutative spacetime is the tensor product of $\mathcal{V} \otimes \mathscr{H}$, where $\mathscr{H}$ is the Hilbert space of the quantum field $\phi$ under consideration. It was soon recognized in Ref.~\onlinecite{GL1} (see also Ref.~\onlinecite{GL2}) that there exists a unitary operator mapping the tensor product space $\mathcal{V} \otimes \mathscr{H}$ to the Hilbert space $\mathscr{H}$. Not only made this mapping computations with the quantum fields easier, but it could be used to prove a weakened locality property of the deformed quantum fields $\phi_\Theta$ (the quantum fields in the non-commutative spacetime), known as wedge-locality, and enabled the non-perturbative calculation of the scattering of two deformed quantum fields.

The next advance in the area of strict deformations of QFT was the realization of warped convolutions~\cite{BS,BLS}. These are a well-defined mathematical prescription on how to obtain the deformed quantum field $\phi_\Theta$ from the undeformed one $\phi$, namely as
\begin{equation}
\phi_\Theta = \int \tau_{\Theta x}(\phi) \total E(x) \eqend{,}
\end{equation}
where $\total E$ is the spectral measure with respect to the translations on $\mathbb{R}^n$, and $\tau$ denotes the action of $\mathbb{R}^n$ on functions by translation: $\tau_a \phi(x) = \phi(x+a)$. It was then shown~\cite{BLS} that this prescription is connected to the Rieffel product~\cite{R}, a noncommutative associative product of functions given by
\begin{equation}
\label{eq:rfp}
f \times_\Theta g = (2\pi)^{-n} \iint \tau_{\Theta x}(f) \tau_y(g) \mathe^{\mathi x \cdot y} \total^n x \total^n y \eqend{,}
\end{equation}
where the oscillatory integral is well-defined for functions $f, g \in \mathcal{D}$ as defined in App.~\ref{oi}. In particular, warped convolutions supply an isometric representation of Rieffel's strict product,
\begin{align}
A_\Theta \, B_\Theta = \left( A \times_\Theta B \right)_\Theta \eqend{,}
\end{align}
for operators $A,B$ that belong to the subclass $\mathcal{C}^{\infty}$ of a $C^*$-algebra which is smooth w.r.t.\ the adjoint action generated by the translation group.

The scope of warped convolutions, however powerful, is unfortunately limited to smooth actions of $\mathbb{R}^n$, i.\,e., mutually commutative generators of the deformation. Hence, to define the deformations one needs a set of commuting operators, at least two. For a curved manifold, commuting Lie vector fields can supply such an action~\cite{Mor}; another possibility would be a spacetime with commuting Killing vector fields~\cite{DFM1}. Yet, as is well known, such vector fields do not exist for arbitrary curved spacetimes, which prevents the application of deformations in general curved spacetimes.

A way out is to embed the curved spacetime manifold (or a patch thereof) into a higher dimensional flat (Minkowski) spacetime. To illustrate the concept, we consider $n$-dimensional de Sitter spacetime, which is embedded into a $(n+1)$-dimensional Minkowski spacetime $\mathbb{R}^{1,n}$ in the standard straightforward way. One obtains in this way an $n$-dimensional group of translations that allow a straightforward definition of the Rieffel product and warped convolutions. Apart from this well-known and comparatively simple embedding, we choose de Sitter spacetime for its importance in early-universe cosmology. In particular, de Sitter spacetime is a good description of the inflationary period in cosmology shortly after the Big Bang (see for example Ref.~\onlinecite[Ch.~11]{cos06} or Ref.~\onlinecite{MFBcosmo}), where the universe expands exponentially and one expects that the quantum nature of spacetime becomes important. The expansion rate is given by the Hubble parameter $H$, and in the limit $H \to 0$ the de Sitter spacetime becomes $n$-dimensional Minkowski spacetime $\mathbb{R}^{1,n-1}$. We show that in this limit, our definition also recovers the non-commutative Minkowski spacetime~\eqref{eq:commutator_mink}, i.e., the deformation and the limit $H \to 0$ commute.

The remainder of this article is structured as follows: in section~\ref{sec:deform_ds}, we review the embedding formalism for de Sitter, define the deformed (Rieffel) product and determine the resulting non-commutative spacetime. In section~\ref{sec:deform_qft}, we determine the corresponding deformed quantum field theory of a free scalar field, and compute the two-point function of the deformed quantum field $\phi_\Theta$. In contrast to Minkowski spacetime, where the deformed two-point function coincides with the undeformed one, we show that in de Sitter spacetime the two differ. We conclude in section~\ref{sec:conclusion}, and collect various proofs in the appendices.

\begin{assumption}
In the following we assume that all functions that we consider are elements of the domain $\mathcal{D}$ defined in App.~\ref{oi}.
\end{assumption}

\section{Deforming de Sitter space}
\label{sec:deform_ds}

To define the Rieffel product for functions on a manifold $M$, one needs an action of $\mathbb{R}^n$. There are various ways to achieve such an action, for example using commuting Killing vector fields, which however do not exist for a general curved spacetime. Therefore, we propose to embed the spacetime under consideration into a higher-dimensional Minkowski spacetime, and deform the action of the (commutative) translations along the embedding coordinates, hence an action of $\mathbb{R}^n$ with which one can define a Rieffel product. In our case, $n$-dimensional de Sitter spacetime can be embedded in $(n+1)$-dimensional Minkowski spacetime $\mathbb{R}^{1,n}$, which is the embedding with the least dimension possible, i.e., a minimal one.

Given thus $\mathbb{R}^{1,n}$ with Cartesian coordinates $X^A$, $A=0,1,\ldots,n$ and the flat metric $\eta_{AB}$, de Sitter spacetime is defined as the hyperboloid
\begin{equation}
\label{eq:desitter_embedding}
\eta_{AB} X^A X^B = \frac{1}{H^2} \eqend{,}
\end{equation}
where $H$ is the Hubble parameter, a positive constant, and we use the Einstein summation convention throughout the paper. The part of de~Sitter spacetime that is relevant in cosmology is known as the (flat, expanding) Poincar{\'e} patch~\cite{MFBcosmo}, and has cosmological time $t$ and spatial Cartesian coordinates $\vec{x}^a$, $a=1,\ldots,n-1$ as coordinates. Those parametrize the hyperboloid~\eqref{eq:desitter_embedding} according to
\begin{equations}[eq:popa]
X^0 &= \frac{1}{H} \sinh( H t ) + \frac{H}{2} \exp( H t ) \vec{x}^2 \eqend{,} \\
X^a &= \exp( H t ) \vec{x}^a \eqend{,} \\
X^n &= \frac{1}{H} \cosh( H t ) - \frac{H}{2} \exp( H t ) \vec{x}^2 \eqend{,}
\end{equations}
where $\vec{x}^2 = \delta_{ab} \vec{x}^a \vec{x}^b$.\footnote{This parametrization can be characterized by $X^0 + X^n > 0$, which covers half of the hyperboloid. The other half, known as the contracting Poincar{\'e} patch, is obtained by inverting the sign of $X^n$.} The induced metric reads
\begin{equation}
\label{eq:stmet1}
g = - \total t^2 + \exp\left( 2 H t \right) \total \vec{x}^2 \eqend{,}
\end{equation}
and as $H \to 0$, we have $X^0 \to t$, $X^a \to \vec{x}^a$ and $g \to \eta$, such that the flat Minkowski space is recovered. In the following, we will consider exclusively the Poincar{\'e} patch, but continue to call it de~Sitter for brevity, in accordance with standard practice in cosmology.

By inverting the parametrization~\eqref{eq:popa}, which results in
\begin{equation}
\label{eq:popb}
t = \frac{1}{H} \ln \left[ H (X^0 + X^n) \right] \eqend{,} \qquad \vec{x}^a = \frac{X^a}{H (X^0 + X^n)} \eqend{,}
\end{equation}
functions $f(x)$ defined on de~Sitter spacetime can be extended to the full embedding spacetime $\mathbb{R}^{1,n}$. We will make extensive use of this extension in the following modified form:
\begin{definition}
\label{def:extension}
Given a function $f$ defined on de~Sitter spacetime, we denote by $f(X)$ its extension to the embedding spacetime using the inverse parametrization~\eqref{eq:popb}, with $X^n$ replaced by $\sqrt{ H^{-2} + (X^0)^2 - \delta_{ab} X^a X^b}$. That is, $f(X)$ is the (uniquely determined) function on $\mathbb{R}^{1,n}$ which is independent of $X^n$ and agrees with the original function $f$ when restricted to the Poincar{\'e} patch~\eqref{eq:popa}.
\end{definition}
To shorten the formulas we will nevertheless continue to write $X^n$, with the implicit understanding that it is just a shorthand for $\sqrt{ H^{-2} + (X^0)^2 - \delta_{ab} X^a X^b }$. We note that this extension is of course not unique, but very convenient for our purposes.

\subsection{The deformed product}
\label{sec_deformed}

In the following we define the deformed product using the embedding formalism. The definition is analogous to the case where a smooth action of the group $\mathbb{R}^n$ exists, i.e., {\`a} la Rieffel~\cite{R}. In our case, the smooth action $\tau$ of $\mathbb{R}^n$ acts as a translation in the embedding coordinates on functions extended to the embedding space.

\begin{definition}
\label{def:defpro}
Let the smooth action $\tau$ of the group $\mathbb{R}^n$ denote translations w.r.t.\ the embedding coordinates $X^\mu = (X^0,X^a)$. Moreover, let $\Theta$ be a constant $n \times n$ skew-symmetric matrix. The non-commutative product between two functions $f,g \in \mathcal{D}$ is defined in analogy to Refs.~\onlinecite{R,BLS} by
\begin{splitequation}
\label{eq:defpro}
\left( f \times_\Theta g \right)(z) &\equiv (2\pi)^{-n} \lim_{\epsilon \to 0^+} \iint \chi(\epsilon X, \epsilon Y) \, \tau_{\Theta X}(f(Z)) \, \tau_{Y}(g(Z)) \, \mathe^{\mathi X \cdot Y} \total^n X \total^n Y \\
&= (2\pi)^{-n} \lim_{\epsilon \to 0^+} \iint \chi(\epsilon X, \epsilon Y) \, f(Z+\Theta X) \, g(Z+Y) \, \mathe^{\mathi X \cdot Y} \total^n X \total^n Y \eqend{,}
\end{splitequation}
where $Z$ is the embedding point corresponding to $z$ according to Def.~\ref{def:extension}, the scalar product $X \cdot Y$ is w.r.t.\ the flat Minkowski metric $\eta_{\mu\nu}$, where $\mu,\nu=0,\ldots,n-1$, i.e., $X \cdot Y = - X^0 Y^0 + \delta_{ab} X^a Y^b$, and the integrations are w.r.t.\ $X^0$ and $X^a$ (resp. $Y^0$ and $Y^a$) only. Moreover, the cut-off function $\chi \in \mathscr{S}(\mathbb{R}^{2n})$ is arbitrary except for the condition $\chi(0,0) = 1$.
\end{definition}
Note that the action $\tau$ on functions $f$ extended from de Sitter spacetime to the embedding space is well-defined, since the extension is constant along the $X^n$ direction. For the same reason, the deformation matrix $\Theta$ only has size $n \times n$ and does not act on the last embedding coordinate $X^n$ at all. While in embedding coordinates, the action $\tau$ is linear, in intrinsic coordinates it is a highly non-linear operation. By its very definition, the deformed product~\eqref{eq:defpro} strongly depends on the exact choice of the extension of functions to the embedding space, but we will see later that our choice has phenomenological advantages.

Apart from the use of translations w.r.t.\ the embedding coordinates instead of intrinsic coordinates, this deformed product is completely analogous to the Rieffel product. Hence, we have the following result:
\begin{proposition}
\label{prop:associative}
The deformed product given in Def.~\ref{def:defpro} is associative, i.e.,
\begin{equation}
\label{eq:associative}
\left( \left( f \times_\Theta g \right) \times_\Theta h \right)(z) = \left( f \times_\Theta	\left( g \times_\Theta h \right) \right)(z)
\end{equation}
for all functions $f, g, h \in \mathcal{D}$ and fulfills
\begin{enumerate}
\item[(i)] $\displaystyle\lim_{\Theta \rightarrow 0} f \times_\Theta g = f g$ \eqend{,}
\item[(ii)] $f \times_\Theta 1 = 1 \times_\Theta f = f$, where $1$ is the constant function equal to $1$.
\end{enumerate}
\end{proposition}
\begin{proof}
The proof is identical to existing proofs of the Rieffel product using smooth actions of $\mathbb{R}^n$, see Ref.~\onlinecite[Theorem 2.14]{R} for the associativity condition~\eqref{eq:associative}, Corollary 2.8 for the limit of vanishing $\Theta$~\ref{prop:associative}~(i) and Corollary 2.13 for the unit element condition~\ref{prop:associative}~(ii).
\end{proof}
Moreover, in the limit $H \to 0$ it becomes the flat-space Rieffel product:
\begin{proposition}
\label{prop:limit}
If the limit $H \to 0$ exists for functions $f$ and $g$, their deformed product turns into the Rieffel product in that limit.
\end{proposition}
\begin{proof}
Choosing a cutoff function $\chi$ that is independent of $H$, the proposition follows from the fact that as $H \to 0$ we have $X^0 \to t$ and $X^a \to \vec{x}^a$~\eqref{eq:popa} and that nothing in the definition~\eqref{eq:defpro} of the deformed product depends on $X^n$, such that the integrand has a smooth pointwise limit. For fixed $\epsilon$ we can interchange the limit $H \to 0$ and the integration in~\eqref{eq:defpro} because of the cutoff function, and the Rieffel product is the limit $\epsilon \to 0$ of the resulting integral.
\end{proof}

\subsection{The resulting non-commutative spacetime}

In the case of deformations on $\mathbb{R}^n$ or the Lorentzian analogue thereof $\mathbb{R}^{1,n-1}$, it is well-known that the Rieffel product leads to non-trivial but simple commutation relations between the coordinate functions. In particular, it leads to a constant non-commutative spacetime with commutation relations
\begin{equation}
\left[ x^\mu, x^\nu \right]_\Theta = 2 \mathi \Theta^{\mu\nu} \eqend{,}
\end{equation}
where we defined
\begin{definition}
The deformed commutator $[\cdot,\cdot]_\Theta$ between two functions $f,g$ is given by the commutator between those two functions with the deformed instead of the pointwise product:
\begin{equation}
\left[ f, g \right]_\Theta \equiv f \times_\Theta g - g \times_\Theta f \eqend{.}
\end{equation}
\end{definition}
For more details about the commutation relations on Minkowski space see Ref.~\onlinecite[Lemma~5.2]{MUc}. Since we have defined the deformed product completely analogous but using the embedding coordinates, we have the same result for the $n$ embedding coordinates $Z^\mu$ and $Z^\nu$:
\begin{equation}
\label{eq:mwst}
\left[ Z^\mu, Z^\nu \right]_\Theta = 2 \mathi \Theta^{\mu\nu} \eqend{.}
\end{equation}
Using that
\begin{equation}
\label{eq:integral_delta}
(2\pi)^{-n} \lim_{\epsilon \to 0^+} \iint \chi(\epsilon X, \epsilon Y) \, f(X) \, \mathe^{ \mathi X \cdot Y } \total^n X \total^n Y = f(0) \eqend{,}
\end{equation}
equation~\eqref{eq:mwst} follows from a straightforward computation:
\begin{splitequation}
\left( Z^\mu \times_\Theta Z^\nu \right)(z) &= (2\pi)^{-n} \lim_{\epsilon \to 0^+} \iint \chi(\epsilon X, \epsilon Y) \, (Z+\Theta X)^\mu \, (Z+Y)^\nu \, \mathe^{ \mathi X \cdot Y } \total^n X \total^n Y \\
&= Z^\mu Z^\nu - \mathi (2\pi)^{-n} \lim_{\epsilon \to 0^+} \iint \chi(\epsilon X, \epsilon Y) \, (\Theta X)^\mu \, \partial_{X^\nu} \, \mathe^{ \mathi X \cdot Y } \total^n X \total^n Y \\
&= Z^\mu Z^\nu + \mathi \Theta^{\mu\nu} \eqend{,}
\end{splitequation}
where integration by parts was permitted because $\chi \in \mathscr{S}(\mathbb{R}^n)$, and the term with the derivative acting on $\chi$ vanishes in the limit $\epsilon \to 0$. Due to the skew symmetry of $\Theta$, one then obtains the commutator relation~\eqref{eq:mwst}.

However, since the embedding spacetime is only an auxiliary construct, we are more interested in the commutation relations that this new deformed product induces from the perspective of the intrinsic de Sitter coordinates. We obtain the following result:
\begin{theorem}
\label{thm:crds}
The deformed commutator relations of the de Sitter spacetime are (to first order in $\Theta$) given by
\begin{equations}[eq:tx_xx_comm]
\left[ t, \vec{x}^b \right]_\Theta &= 2 \mathi \mathe^{- 2 H t} \left[ \Theta^{0b} - \Theta^{\rho b} F_\rho \right] + \bigo{\Theta^2} \eqend{,} \\
\left[ \vec{x}^a, \vec{x}^b \right]_\Theta &= 2 \mathi \mathe^{- 2 H t} \left[ \Theta^{ab} -  H\,\vec{x}^a \left( \Theta^{0b} - \Theta^{\rho b} F_\rho \right) +  H\,\vec{x}^b \left( \Theta^{0a} - \Theta^{\rho a} F_\rho \right) \right] + \bigo{\Theta^2} \eqend{,}
\end{equations}
where $F$ is the $n$-dimensional vector
\begin{equation}
F^\mu \equiv \frac{X^\mu}{X^n} = \left( \cosh( H t ) - \frac{H^2}{2} \exp( H t ) \vec{x}^2 \right)^{-1} \begin{pmatrix} \sinh( H t ) + \frac{H^2}{2} \exp( H t ) \vec{x}^2 \\ H \exp( H t ) \vec{x}^a \end{pmatrix} \eqend{.}
\end{equation}
\end{theorem}
\begin{proof}
We first derive a formally equivalent representation of the deformed Rieffel product~\ref{def:defpro} in terms of differential operators, which is useful for an expansion in powers of the non-commutativity matrix $\Theta$. For this, we expand $f(Z + \Theta X)$ in equation~\eqref{eq:defpro} in a Taylor series in $\Theta$, which gives
\begin{splitequation}
&\left( f \times_\Theta g \right)(z) = \sum_{k=0}^\infty \frac{1}{k!} (2\pi)^{-n} \lim_{\epsilon \to 0^+} \iint \chi(\epsilon X, \epsilon Y) \, (\Theta X)^{\otimes k} \left( \nabla^k f \right)(Z) \, g(Z+Y) \, \mathe^{\mathi X \cdot Y} \total^n X \total^n Y \\
&\qquad= \sum_{k=0}^\infty \frac{(-\mathi)^k}{k!} \left( \nabla^k f \right)(Z) (2\pi)^{-n} \lim_{\epsilon \to 0^+} \iint \chi(\epsilon X, \epsilon Y) \, g(Z+Y) \, (\Theta \nabla_Y)^{\otimes k} \mathe^{\mathi X \cdot Y} \total^n X \total^n Y \eqend{,}
\end{splitequation}
where the sum over $k$ is to be understood in the sense of formal power series in $\Theta$. We then integrate the derivatives with respect to $Y$ by parts which is permissible again because $\chi \in \mathscr{S}(\mathbb{R}^{2n})$, and note that derivatives acting on $\chi$ result in powers of $\epsilon$ and hence vanish in the limit $\epsilon \to 0$. Using then equation~\eqref{eq:integral_delta}, we obtain
\begin{splitequation}
\label{eq:product_exp_derivative}
\left( f \times_\Theta g \right)(z) &= \sum_{k=0}^\infty \frac{\mathi^k}{k!} \left( \nabla^k f \right)(Z) \left( (\Theta \nabla)^{\otimes k} g \right)(Z) \\
&= \Big[ \exp\!\big( \mathi \Theta^{\mu\nu} \partial_{X^\mu} \partial_{Y^\nu} \big) f(X) g(Y) \Big]_{X = Y = Z} \eqend{,}
\end{splitequation}
which is valid as formal power series in $\Theta$. It becomes non-formal (i.e, strict), if the series of derivatives converges, i.e., if $f$ and $g$ are analytic functions. Note however that the original definition~\ref{def:defpro} is valid for all functions $f,g \in \mathcal{D}$.

We then use the inverse parametrization~\eqref{eq:popb} to compute
\begin{splitequation}
\left[ \vec{x}^a, \vec{x}^b \right]_\Theta &= \vec{x}^a \vec{x}^b + \mathi \left[ \big( \partial_{X^\mu} \vec{x}^a \big) \Theta^{\mu\nu} \big( \partial_{Y^\nu} \vec{y}^b \big) \right]_{X = Y} + \bigo{\Theta^2} - (a \leftrightarrow b) \\
&= \mathi \left[ \frac{1}{H (X^0 + X^n)} \delta^a_\mu - \frac{X^a}{H (X^0 + X^n)^2} \left( \delta^0_\mu - \frac{X_\mu}{X^n} \right) \right] \Theta^{\mu\nu} \\
&\qquad\times \left[ \frac{1}{H (X^0 + X^n)} \delta^b_\nu - \frac{X^b}{H (X^0 + X^n)^2} \left( \delta^0_\nu - \frac{X_\nu}{X^n} \right) \right] - (a \leftrightarrow b) + \bigo{\Theta^2} \\
\end{splitequation}
where we recall that $X^n$ is a shorthand for $\sqrt{ H^{-2} + (X^0)^2 - \delta_{ab} X^a X^b} = \sqrt{ H^{-2} - \eta_{\mu\nu} X^\mu X^\nu }$ according to Def.~\ref{def:extension}. Using that from the relations~\eqref{eq:popa} it follows that $H (X^0 + X^n) = \exp(H t)$ and $X^a = \exp(H t) \vec{x}^a$, the result follows. In the same way, we obtain
\begin{splitequation}
\left[ t, \vec{x}^a \right]_\Theta &= \mathi \left[ \big( \partial_{X^\mu} t \big) \Theta^{\mu\nu} \big( \partial_{Y^\nu} \vec{y}^a \big) - \left( X \leftrightarrow Y \right) \right]_{X = Y} + \bigo{\Theta^2} \\
&= \mathi \bigg[ \frac{1}{H (X^0 + X^n)} \left( \delta^0_\mu - \frac{X_\mu}{X^n} \right) \Theta^{\mu\nu} \left[ \frac{1}{H (Y^0 + Y^n)} \delta^a_\nu - \frac{Y^a}{H (Y^0 + Y^n)^2} \left( \delta^0_\nu - \frac{Y_\nu}{Y^n} \right) \right] \\
&\qquad\qquad- \left( X \leftrightarrow Y \right) \bigg]_{X = Y} + \bigo{\Theta^2} \eqend{,}
\end{splitequation}
which with $H (X^0 + X^n) = \exp(H t)$ and $X^a = \exp(H t) \vec{x}^a$ gives the stated result.
\end{proof}
\begin{corollary}
The resulting non-commutative spacetime given in Theorem~\ref{thm:crds} becomes the Moyal--Weyl spacetime in the flat limit $H \to 0$, i.e.,
\begin{equation}
\lim_{H \to 0} \left[ x^\mu, x^\nu \right]_\Theta = 2 \mathi \Theta^{\mu\nu} \eqend{.}
\end{equation}
\end{corollary}
\begin{proof}
To first order in $\Theta$, this can be seen directly from Eqs.~\eqref{eq:tx_xx_comm}, and in general it follows from Proposition~\ref{prop:limit}.
\end{proof}

In particular, since deformation of the Minkowski spacetime using the Rieffel product leads to the non-commutative Moyal-Weyl spacetime, this shows that the deformation and the limit commute, which we display in the commuting diagram of Fig.~\ref{pic:commuting}.
\begin{figure}
\includegraphics{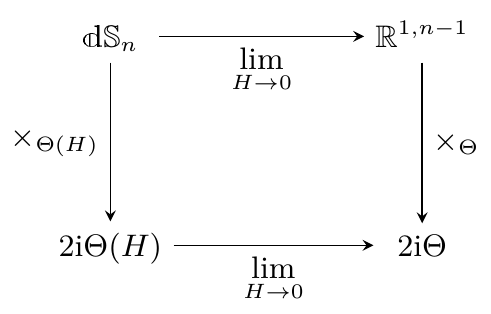}
\caption{Commuting diagram for deformations of de~Sitter and Minkowski spacetime. Here, $\Theta(H)$ represents the deformed commutator relations given in Theorem~\ref{thm:crds}, and $\mathbbm{d}\mathbb{S}_n$ is the (expanding Poincar{\'e} patch) of de Sitter spacetime.}
\label{pic:commuting}
\end{figure}

At this point, we want to mention that the deformed commutator relations fulfill the Jacobi identity as a direct consequence of the associativity of the deformed product, Proposition~\ref{prop:associative}.

\subsection{Physical Interpretation}

Several motivations, both mathematical and physical, lead to the assumption of a non-commutative spacetime. One physical argument, for example, is the geometrical measurement problem~\cite{Ahluwalia1993,DFR} as explained in the introduction. An example where a non-commutative space is physically realized is the Landau quantization. This physical situation occurs for non-relativistic electrons confined to a plane in the presence of a homogenous magnetic field, see Ref.~\onlinecite{Muc1}.

If one, however, assumes that a non-commutative spacetime, a quantum spacetime, can resolve the Big Bang singularity, it is clear that a deformed flat spacetime (of Moyal--Weyl type) is not a satisfying model. Instead, one needs to consider curved (or generically dynamical) non-commutative spacetimes, such as the one we have obtained in Theorem~\ref{thm:crds}. Moreover, a striking feature of the non-commutativity has appeared naturally: as time passes and the universe expands, the non-commutativity dynamically vanishes. After the end of inflation, the non-commutativity is completely negligible even if at its beginning it is of order one, since the overall multiplying factor $\exp(- H t)$ decreases at least by $\mathe^{-60}$ from the beginning to the end of inflation, see Ref.~\onlinecite[Ch.~12]{cos06}. Since the assumption of a constant non-commutative spacetime, even for very small non-commutativity, leads to violations of Lorentz symmetry which have not been observed~\cite{LVB16}, such models are unsustainable; even if one fundamentally violates Lorentz symmetry only at the Planck scale, quantum fluctuations of the metric enhance this breaking also for low energies~\cite{Knorr2018}. In contrast, the dynamical non-commutativity that results in our proposal can easily accommodate no detectable violations of Lorentz symmetry during most of the cosmological evolution including today, and only modifies the inflationary period and the Big Bang singularity. Concrete predictions and results from this model are under investigation.

\section{Deforming quantum field theory}
\label{sec:deform_qft}

Rigorous deformations of quantum field theory in the Minkowski spacetime have been initiated by Ref.~\onlinecite{GL1,GL2} and connected to the Rieffel product in Ref.~\onlinecite{BLS}, see also Refs.~\onlinecite{GL4,GL5,MUc} and the reviews~\onlinecite{douglasnekrasovreview,SZNCQFT}. However, for deformations of QFT in Minkowski spacetime the two-point function of the deformed quantum field $\phi_\Theta$ is identical to the one of the undeformed quantum field $\phi$. Since the measured scalar power spectrum of the cosmic microwave background is directly related to the two-point function of the metric perturbations~\cite{MFBcosmo}, this deformation thus has no observable effects, which is is quite unsatisfactory. As we will see in this section, by deforming the QFT in de Sitter spacetime using the deformed Rieffel product defined in section~\ref{sec_deformed}, we obtain a deformed non-commutative model whose two-point function differs from the undeformed one, and which could therefore give observable effects.

To construct the deformed model, we follow the same steps of the construction as in Minkowski spacetime~\cite{GL1, GL2, MUc}. That is, we first construct the undeformed QFT using canonical quantization. In a second step, we need to define a unitary group that implements the deformation action of $\mathbb{R}^n$ on the quantum fields, which we do by constructing the infinitesimal generators of translations and exponentiating them. In Minkowski space, these are just the global momentum operators, but we will see that in our case they are not global but depend on the point where they act. Lastly, we define the deformed quantum field by taking our deformed Rieffel product between the quantum field and the elements of the unitary group, and compute its two-point function.

\subsection{Canonical quantization}

For simplicity, we consider a free massless and conformally coupled field $\phi$ that fulfills the Klein--Gordon equation:
\begin{equation}
\left( \nabla^2 - \frac{n (n-2)}{4} H^2 \right) \phi(x) = 0 \eqend{,}
\end{equation}
where the term proportional to $H^2$ results from the conformal coupling.\cite{AllenDeSitter} In the Poincar{\'e} patch with metric~\eqref{eq:stmet1}, this reads
\begin{equation}
\label{eq:phi_kg}
\left[ - \partial_t^2 - (n-1) H \partial_t + \exp (- 2 H t) \laplace - \frac{n (n-2)}{4} H^2 \right] \phi(t,\vec{x}) = 0 \eqend{,}
\end{equation}
where $\laplace \equiv \delta^{ab} \partial_{\vec{x}^a} \partial_{\vec{x}^b}$. The canonical momentum is given by
\begin{equation}
\pi = \sqrt{-\det g} \partial_t \phi = \exp\left[ (n-1) H t \right] \partial_t \phi \eqend{,}
\end{equation}
and following canonical quantization, the mode expansion of the field operator reads
\begin{equation}
\label{eq:qf}
\phi(\vec{x},t) = \int \left[ a(\vec{p}) f(\vec{p},t) + a^\dagger(-\vec{p}) f^*(-\vec{p},t) \right] \mathe^{\mathi \vec{p} \vec{x}} \frac{\total^{n-1} p}{(2\pi)^{n-1}} \eqend{,}
\end{equation}
with creation and annihilation operators fulfilling the commutation relation
\begin{equation}
\label{eq:caop}
\left[ a(\vec{p}), a^\dagger(\vec{q}) \right] = (2\pi)^{n-1} \delta^{n-1}(\vec{p}-\vec{q}) \eqend{.}
\end{equation}
With the mode functions normalized to
\begin{equation}
\label{eq:phi_mode_norm}
f(\vec{p},t) \partial_t f^*(\vec{p},t) - f^*(-\vec{p},t) \partial_t f(-\vec{p},t) = \mathi \exp\left[ - (n-1) H t \right] \eqend{,}
\end{equation}
we have the canonical commutation relations
\begin{equation}
\left[ \phi(\vec{x},t), \pi(\vec{y},t) \right] = \mathi \delta^{n-1}(\vec{x}-\vec{y}) \eqend{.}
\end{equation}
Plugging the expansion~\eqref{eq:qf} into the Klein--Gordon equation~\eqref{eq:phi_kg}, the mode functions are seen to fulfill
\begin{equation}
\left[ \partial_t^2 + (n-1) H \partial_t + \exp(- 2 H t) \vec{p}^2 + \frac{n (n-2)}{4} H^2 \right] f(\vec{p},t) = 0 \eqend{,}
\end{equation}
which has the positive-frequency solution (normalized according to equation~\eqref{eq:phi_mode_norm})
\begin{splitequation}
f(\vec{p},t) &= \frac{1}{\sqrt{2 \abs{\vec{p}}}} \exp\left[ - \frac{n-2}{2} H t + \mathi \big[ \exp( - H t ) - 1 \big] \frac{\abs{\vec{p}}}{H} \right] \\
&\to \frac{1}{\sqrt{2 \abs{\vec{p}}}} \exp\left( - \mathi \abs{\vec{p}} t \right) \quad (H \to 0) \eqend{.}
\end{splitequation}
The operator $a(\vec{p})$ annihilates the vacuum vector $\ket{0}$, which with the above choice of mode functions is the Bunch--Davies vacuum~\cite{BD}. By standard methods we obtain the one-particle Hilbert space $\mathscr{H}$\footnote{One way is to apply the creation operator on the vacuum vector to create one-particle states, and construct the one-particle Hilbert space as completion of their linear span. Another way to obtain this space is to construct it explicitly by using the symplectic structure and a complex structure, as in Ref.~\onlinecite{MO}.}, from which we the define the bosonic Fock space $\mathscr{F}$ as the completion of $\mathscr{F}_0 = \displaystyle\bigoplus_{k=0}^{\infty} \mathscr{H}_k$, i.e., $\mathscr{F} \equiv \overline{\mathscr{F}_0}$  with $\mathscr{H}_0 = \mathbb{C}$ and $\mathscr{H}_k = \mathscr{H}^{\otimes_S k}$, where $\otimes_S$ denotes the symmetric tensor product.

\subsection{The deformed quantum field}

Since our deformed product is obtained from the action of $\mathbb{R}^n$ in the embedding space, we now construct the operators that implement infinitesimal translations in embedding space. Because these translations commute, we obtain in this way $n$ linearly independent commuting operators $P_\mu$ acting as
\begin{equation}
\left[ P_\mu(X), \phi(X) \right] = - \mathi \partial_{X^\mu} \phi(X) \eqend{,}
\end{equation}
where we anticipated that they depend on the point where they act. Here we recall that $\phi(X)$ denotes the extension of $\phi(x)$ [defined by the expansion~\eqref{eq:qf}] to the embedding space according to Def.~\ref{def:extension}. From the inverse parametrization~\eqref{eq:popb}, we compute
\begin{equations}
\frac{\partial t}{\partial X^0} &= \frac{1}{H X^n} \eqend{,} \\
\frac{\partial \vec{x}^a}{\partial X^0} &= - \frac{\vec{x}^a}{X^n} \eqend{,} \\
\frac{\partial t}{\partial X^b} &= - \frac{X_b}{H X^n (X^0 + X^n)} = - \frac{\vec{x}_b}{X^n} \eqend{,} \\
\frac{\partial \vec{x}^a}{\partial X^b} &= \frac{1}{H (X^0 + X^n)} \delta^a{}_b + \frac{H \vec{x}^a \vec{x}_b}{X^n} \eqend{,}
\end{equations}
where we recall that $X^n = \sqrt{ H^{-2} + (X^0)^2 - \delta_{ab} X^a X^b }$, and obtain in intrinsic coordinates
\begin{equation}
\label{eq:p0phi_comm}
\left[ P_0(X), \phi(X) \right] = - \mathi \frac{\partial x^\mu}{\partial X^0} \frac{\partial}{\partial x^\mu} \phi(X) = - \frac{\mathi}{H X^n} \left[ \partial_t \phi(x) - H \vec{x}^b \partial_{x^b} \phi(x) \right]
\end{equation}
and
\begin{equation}
\label{eq:piphi_comm}
\left[ P_a(X), \phi(X) \right] = - \frac{\mathi}{H X^n} \left[ \frac{X^n}{X^0 + X^n} \partial_{\vec{x}^a} \phi(x) + H^2 \vec{x}_a \vec{x}^b \partial_{\vec{x}^b} \phi(x) - H \vec{x}_a \partial_t \phi(x) \right] \eqend{.}
\end{equation}
As $H \to 0$, we have $H X^n \to 1$ and $H X^0 \to 0$~\eqref{eq:popa}, and recover the correct flat limit. To determine the actual form of the $P_\mu(X)$, we use the expansion~\eqref{eq:qf}, make an ansatz quadratic in creation and annihilation operators and   evaluate the unknown coefficients by requiring~\eqref{eq:p0phi_comm} and~\eqref{eq:piphi_comm} to hold. The derivation and result are given in App.~\ref{sec:caop}. By standard arguments~\cite{RS2}, the so-defined operators $P_\mu$ are essentially self-adjoint operators on a dense domain $\mathcal{D}_1 \subset \mathscr{F}$ for each fixed point $X$, and thus in general. From the explicit result of App.~\ref{sec:caop}, the $P_\mu$ do not annihilate the vacuum vector $\ket{0}$, contrary to the Minkowski case (even though their expectation value vanishes). We conclude that they do \emph{not} coincide with the standard momentum operators that can be defined in de Sitter QFT~\cite{DS1, DS2, DS3} and which do annihilate the vacuum vector. However, for us this is a good result and the essential reason why we are able to obtain a two-point function that differs from the undeformed one in the following. From the operators $P_\mu(X)$, one obtains the unitary group implementing finite translations again in the standard way:
\begin{definition}
\label{def:unitary}
For all points $X^\mu = (X^0, X^a)$, $Y^\mu$ and $Z^\mu$ in the embedding space, a two-parameter group of unitary operators is uniquely defined by
\begin{equation}
U(X,Y) = U(Y,X)^\dagger \eqend{,} \quad U(X,X) = \1 \eqend{,} \quad \partial_{X^\mu} U(X,Y) = - \mathi P_\mu(X) \, U(X,Y) \eqend{.}
\end{equation}
\end{definition}

Completely analogous to the Minkowski case~\cite{GL1,GL2,BLS}, we then define the deformed field operators:
\begin{definition}
\label{defqfx}
Let $\Theta$, $\chi$ and the dot product be as in the definition of the deformed product~\ref{def:defpro}, let $\phi(x)$ be the free scalar field~\eqref{eq:qf}, and let $\tau$ denote the automorphism defined by the adjoint action of the unitary group defined in~\ref{def:unitary}:
\begin{equation}
\label{eq:tau_phi_def}
\tau_Y \phi(X) = U(X+Y,X) \phi(X) U(X,X+Y) = \phi(X+Y) \eqend{.}
\end{equation}
The deformed quantum field $\phi_\Theta(x)$ is then formally defined by warped convolution as
\begin{equation}
\label{defqfx2}
\phi_\Theta(z) \equiv (2\pi)^{-n} \lim_{\epsilon \to 0^+} \iint \chi(\epsilon X,\epsilon Y) \, \tau_{\Theta X}\left( \phi(Z) \right) \, U(Z+Y,Z) \, \mathe^{\mathi X \cdot Y} \total^n X \total^n Y \eqend{,}
\end{equation}
whenever the limit exists, and the deformed field ${}_\Theta\phi(x)$ is defined by
\begin{equation}
\label{defqfx3}
{}_\Theta\phi(z) \equiv (2\pi)^{-n} \lim_{\epsilon \to 0^+} \iint \chi(\epsilon X,\epsilon Y) \, U(Z,Z+Y) \, \tau_{\Theta X}\left( \phi(Z) \right) \, \mathe^{\mathi X \cdot Y} \total^n X \total^n Y \eqend{.}
\end{equation}
In order to define the field (an operator-valued distribution) rigorously, one has to smear it with a real Schwartz function $f \in \mathscr{S}(\mathbb{R}^n)$\footnote{The reality condition ensures the self-adjointness of the scalar field.}, and we denote the smeared field by $\phi_\Theta(f)$.
\end{definition}
It was proven in Ref.~\onlinecite[Lemma 2.2]{BS} that the two definitions $\phi_\Theta$ and ${}_\Theta\phi$ agree in Minkowski space, and the proof can be taken over. Moreover, comparing with the definition of the deformed product~\ref{def:defpro} one has formally
\begin{equation}
\phi_\Theta(z) = \left( \phi \times_\Theta U(\cdot,z) \right)(z) \eqend{,} \qquad {}_\Theta\phi(z) = \left( U(z,\cdot) \times_\Theta \phi \right)(z) \eqend{.}
\end{equation}

Furthermore, in analogy with Refs.~\onlinecite[Prop. 2.2 b)]{GL1} and~\onlinecite[Prop. 3.1]{Muc1} we want to prove that the (smeared) deformed operators $\phi_\Theta(f)$ are well-defined operators in the Fock space~$\mathscr{F}$:
\begin{proposition}
Let us consider the smeared deformed field $\phi_{\Theta}(f)$, and let $\Phi, \Psi \in \mathscr{F}$ be such that $\norm{ P_\mu(X_1) \cdots P_\mu(X_k) \Phi } < \infty$ for all $k \leq 2n+2$. Then the scalar product $\expect{ \Phi, {}_\Theta\phi(f) \Psi } = \expect{ \Phi, \phi_\Theta(f) \Psi }$ is bounded. Hence the deformation of the free field in de Sitter spacetime, given by the deformed field $\phi_{\Theta}(f)$ of Def.~\ref{defqfx} is well-defined.
\end{proposition}
\begin{proof}
We have
\begin{equation}
\label{eq:thm_scalarprod} 
\expect{ \Phi, {}_\Theta\phi(f) \Psi } = (2\pi)^{-n} \lim_{\epsilon \to 0^+} \iint \chi(\epsilon X,\epsilon Y) \, b_\Theta(X,Y,f) \, \mathe^{\mathi X \cdot Y} \total^n X \total^n Y
\end{equation}
with $b_\Theta(X,Y,f)$ formally given by
\begin{equation}
b_\Theta(X,Y,f) = \int f(Z) \, \expect{ \Phi, U(Z,Z+Y) \, \tau_{\Theta X}\left( \phi(Z) \right) \Psi } \total^n Z \eqend{.}
\end{equation}
By Ref.~\onlinecite[Ineq.~1.4]{R} and choosing a suitable norm, the oscillatory integral in equation~\eqref{eq:thm_scalarprod} converges absolutely and we have for some $k > n$ (e.g., $k = n+1$) the estimate
\begin{equation}
\abs{ \expect{ \Phi, {}_\Theta\phi(f) \Psi } } \leq c_k \sum_{i+j \leq 2k} \sup_{X,Y} \abs{ \nabla^i_X \nabla^j_Y b_\Theta(X,Y,f) }
\end{equation}
for some constant $c_k$. We use a standard trick and decompose the test function as
\begin{equation}
f(Z) = \left[ f(Z) (1 + \delta_{\mu\nu} Z^\mu Z^\nu)^n \right] (1 + \delta_{\mu\nu} Z^\mu Z^\nu)^{-n} \equiv \tilde{f}(Z) g(Z) \eqend{,}
\end{equation}
where $\tilde{f}(Z) \equiv f(Z) (1 + \delta_{\mu\nu} Z^\mu Z^\nu)^n$ is still a Schwartz test function. We then compute for $i = j = 0$ that
\begin{splitequation}
\sup_{X,Y} \abs{ b_\Theta(X,Y,f) } &= \sup_{X,Y} \abs{ \int f(Z) \, \expect{ \Phi, U(Z,Z+Y) \, \tau_{\Theta X}\left( \phi(Z) \right) \Psi } \total^n Z } \\
&= \sup_{X,Y} \abs{ \int \tilde{f}(Z) g(Z) \, \expect{ U(Z,Z+Y)^\dagger \Phi, \tau_{\Theta X}\left( \phi(Z) \right) \Psi } \total^n Z } \\
&\leq \sup_{X,Y} \norm{ U(g,Y)^\dagger \Phi } \norm{ \left( \tau_{\Theta X} \phi \right)(\tilde{f}) \Psi } \eqend{,}
\end{splitequation}
where we used the Cauchy--Schwarz inequality, and we denote by $U(g,Y)^\dagger$ the smearing of $U(Z,Z+Y)^\dagger$ in $Z$ with $g$. Since $U$ is a unitary operator and thus of norm 1, we have
\begin{equation}
\sup_Y \norm{ U(g,Y)^\dagger \Phi } \leq \norm{ g }_1 \norm{ \Phi }
\end{equation}
with $\norm{ g }_1 < \infty$ the $L_1$ norm of $g$ which is finite by construction. Furthermore, we have
\begin{equation}
\sup_X \norm{ \left( \tau_{\Theta X} \phi \right)(\tilde{f}) \Psi } = \sup_X \norm{ \phi\left( \tilde{f} \circ \tau_{-\Theta X} \right) \Psi } \leq c'(\Theta,\tilde{f}) \norm{ \Psi } \eqend{,}
\end{equation}
where we used that (because the smeared field $\phi$ is a well-defined operator in Fock space)
\begin{equation}
\norm{ \phi(f) \Psi } \leq c(f) \norm{ \Psi }
\end{equation}
holds for some positive constant $c$ depending on the Schwartz seminorms $\norm{ f }_{\alpha,\beta}$ of the test function $f$, and that (with $\norm{ \Theta }_\infty$ the maximum modulus of the components of $\Theta$)
\begin{equation}
\sup_X \norm{ \tilde{f} \circ \tau_{-\Theta X} }_{\alpha,\beta} = \sup_{X,Y} \abs{ Y^\beta \nabla^\alpha \tilde{f}(Y-\Theta X) } \leq \norm{ \Theta }_\infty^\beta \norm{ \tilde{f} }_{\alpha,\beta} \eqend{.}
\end{equation}
For derivatives w.r.t.\ $X$, one obtains instead derivatives of $\phi$ and powers of $\Theta$, which have again finite norm [possibly with a different constant $c'(f)$], while for derivatives w.r.t.\ $Y$ acting on $U$ we obtain powers of $P_\mu$. Requiring that the norm of these acting on $\Phi$ be finite, we obtain the stated restriction on the vector $\Phi$ and the conclusion.
\end{proof}

\subsection{Two-point function}

Lastly, we compute the two-point function of the deformed field, up to first order in $\Theta$. For this, we first expand the deformed field to first order in $\Theta$. Using the definition of the unitary group~\ref{def:unitary} and the automorphism $\tau$~\eqref{eq:tau_phi_def}, we obtain
\begin{equation}
\tau_{\Theta X}\left( \phi(Z) \right) = \phi(Z) - \mathi \Theta^{\mu\nu} X_\nu \left[ P_\mu(Z), \phi(Z) \right] + \bigo{\Theta^2} \eqend{.}
\end{equation}
Inserting this result into the definition~\eqref{defqfx2} and using~\eqref{eq:integral_delta}, it follows that
\begin{splitequation}
\label{eq:phitheta_firstorder}
\phi_\Theta(z) &= \phi(Z) - \Theta^{\mu\nu} (2\pi)^{-n} \lim_{\epsilon \to 0^+} \iint \chi(\epsilon X,\epsilon Y) \, \left[ P_\mu(Z), \phi(Z) \right] \, U(Z+Y,Z) \\
&\hspace{12em}\times \partial_{Y^\nu} \mathe^{\mathi X \cdot Y} \total^n X \total^n Y + \bigo{\Theta^2} \\
&= \phi(Z) + \Theta^{\mu\nu} \left[ P_\mu(Z), \phi(Z) \right] \, \partial_{Y^\nu} U(Z+Y,Z) \Big\rvert_{Y=0} + \bigo{\Theta^2} \\
&= \phi(Z) - \mathi \Theta^{\mu\nu} \left[ P_\mu(Z), \phi(Z) \right] P_\nu(Z) + \bigo{\Theta^2} \eqend{.}
\end{splitequation}
The computation of the two-point function is then long but straightforward, expressing $\phi$ and the $P_\mu$ in terms of annihilation and creation operators, acting with them on the vacuum vector $\ket{0}$ and computing the resulting momentum integrals. We obtain
\begin{theorem}
\label{thm:tpf}
The two-point function of the non-commutative scalar $\phi_\Theta$ in the $n$-dimensional de Sitter space reads
\begin{splitequation}
\label{eq:tpf}
\bra{0} \phi_\Theta(x) \phi_\Theta(y) \ket{0} &= \lim_{\epsilon \to 0^+} \Bigg[ \frac{\Gamma\left( \frac{n-2}{2} \right)}{2 (2\pi)^\frac{n}{2}} H^{n-2} \left[ 1 - Z(X,Y) + \mathi \epsilon \sgn(t-s) \right]^\frac{2-n}{2} \\
&\quad+ \mathi \frac{(n-2) \Gamma\left( \frac{n}{2} \right)}{4 (2\pi)^\frac{n}{2}} \left[ \frac{\exp( H s )}{X^n} + \frac{\exp( H t )}{Y^n} \right] \left( \vec{x}_a \Theta^{a0} + \Theta^{0b} \vec{y}_b - H \vec{x}_a \Theta^{ab} \vec{y}_b \right) \\
&\quad\qquad\times H^n \left[ 1 - Z(X,Y) + \mathi \epsilon \sgn(t-s) \right]^{-\frac{n}{2}} \Bigg] + \bigo{\Theta^2} \raisetag{2.1em}
\end{splitequation}
with the de Sitter invariant
\begin{splitequation}
Z(X,Y) &= H^2 \eta_{AB} X^A Y^B = \cos\left( H \mu(x,y) \right) \\
&= \cosh\left[ H (t-s) \right] - \frac{H^2}{2} \exp[ H (t+s) ] (\vec{x}-\vec{y})^2 \eqend{,}
\end{splitequation}
where $\mu(x,y)$ is the geodesic distance between $x$ and $y$ if they are connected by a geodesic.
\end{theorem}
\begin{proof}
See App.~\ref{proof:tpf}.
\end{proof}
For applications in inflationary cosmology, one needs to consider the two-point function at equal times. We have
\begin{corollary}
\label{corr:latetime_tpf}
For $t \to \infty$, the leading behavior of the the two-point function of the non-commutative scalar $\phi_\Theta$ at equal times is given by
\begin{splitequation}
\bra{0} \phi_\Theta(x) \phi_\Theta(y) \ket{0} &\sim \frac{\Gamma\left( \frac{n-2}{2} \right)}{4 \pi^\frac{n}{2}} \left[ (\tilde{\vec{x}}-\tilde{\vec{y}})^2 \right]^\frac{2-n}{2} \\
&\quad+ \mathi \frac{(n-2) \Gamma\left( \frac{n}{2} \right)}{\pi^\frac{n}{2}} H \exp( - H t ) \left( \tilde{\vec{x}}_a \Theta^{a0} + \Theta^{0b} \tilde{\vec{y}}_b \right) + \bigo{\Theta^2} \eqend{,}
\end{splitequation}
where $\tilde{\vec{x}} \equiv \exp( H t ) \vec{x}$ are the physical (co-moving) coordinates~\cite{MFBcosmo}. We see that the corrections coming from the non-commutativity are again exponentially suppressed.
\end{corollary}

\begin{remark}
Assuming that $\Theta^{\mu\nu}$ is independent of $H$, we can also expand the result~\ref{thm:tpf} for small $H$ and obtain up to first order in $\Theta$ and second order in $H$ the following result:
\begin{splitequation}
&\bra{0} \phi_\Theta(X) \phi_\Theta(Y) \ket{0} = \lim_{\epsilon \to 0^+} \Bigg[ \frac{\Gamma\left( \frac{n-2}{2} \right) H^{n-2}}{2 (2\pi)^\frac{n}{2}} \left[ 1 - Z(X,Y) + \mathi \epsilon \sgn(t-s) \right]^\frac{2-n}{2} \\
&\qquad+ \mathi \frac{(n-2) \Gamma\left( \frac{n}{2} \right)}{4 \pi^\frac{n}{2}} H \left[ (x-y)^2 + \mathi \epsilon \sgn(t-s) \right]^{-\frac{n}{2}} \\
&\qquad\qquad\times \biggl[ 2 \vec{x}_a \Theta^{a0} + 2 \Theta^{0b} \vec{y}_b - 2 H \vec{x}_a \Theta^{ab} \vec{y}_b \\
&\qquad\qquad\qquad\quad+ H (t+s) \left( \vec{x}_a \Theta^{a0} + \Theta^{0b} \vec{y}_b \right) \left( 1 - \frac{n (\vec{x}-\vec{y})^2}{(x-y)^2 + \mathi \epsilon \sgn(t-s)} \right) \biggr] \Bigg] \eqend{,}
\end{splitequation}
where $(x-y)^2 \equiv (\vec{x}-\vec{y})^2 - (t-s)^2$. Note that the non-commutativity vanishes for $n = 2$, already before expanding for small $H$, which is due to fact that in this case the operators $P_\mu$ annihilate the vacuum vector $\ket{0}$, analogously to what happens in Minkowski space. Furthermore, one sees very clearly that the non-commutative contributions also vanish in the limit $H \to 0$.
\end{remark}

At this point, we note an important difference to the Minkowski case: even though the the deformation of the quantum fields is defined by the same warped convolution as in Ref.~\onlinecite{BLS}, it does \emph{not} give an isometric representation of the deformed Rieffel product, i.e., in our case
\begin{equation}
\left( A \times_\Theta B \right)_\Theta \neq A_\Theta B_\Theta
\end{equation}
in general. This results from the fact that the infinitesimal generators of translations $P_\mu(X)$ are not global but depend on a point $X$, and do not commute at different points: $[ P_\mu(X), P_\nu(Y) ] \neq 0$ for $X \neq Y$ (see App.~\ref{app:pmu_comm}).

\section{Conclusion and outlook}
\label{sec:conclusion}

We have proposed a new way of deforming quantum field theories on curved spacetimes, by embedding them into a higher-dimensional Minkowski spacetime and constructing a Rieffel product and warped convolutions in the embedding spacetime. In our model, not only does the non-commutativity decrease exponentially as time progresses (such that constraints from Lorentz violation are fulfilled), but also the two-point function of the deformed scalar field differs from the undeformed case, with possible observable consequences. As an extension of the present work, it would be very interesting to determine potentially observable effects during the inflationary period, and to compare with other approaches to non-commutative geometry such as those of  Refs.~\onlinecite{brandenberger2002,huang2003,calcagni2004,oliveiraneto2021}.

Apart from the de Sitter spacetime, there exists a straightforward generalization of our deformation procedure to all $n$-dimensional maximally symmetric space(-times), with the embedding spacetime being Euclidean space (for the sphere), Minkowski spacetime (for hyperbolic space) or $\mathbb{R}^{2,n-1}$ (for Anti-de~Sitter spacetime). For more general curved spacetimes, the question of embedding is somewhat more complicated. It is known by the Nash embedding theorem that every Riemannian manifold (compact or not) can be isometrically embedded in $\mathbb{R}^m$ for some large enough $m$. In Ref.~\onlinecite{MS10} this was extended to certain classes of Lorentzian manifolds (stably causal spacetimes with steep temporal function), which in particular include globally hyperbolic spacetimes in which quantum field theories can be defined rigorously. In this way we can define a strict deformation of quantum field theories in any globally hyperbolic manifold, at least in principle.

While the chosen embedding for de~Sitter spacetime is standard, it is not unique, and we expect the noncommutative spacetime to change in general depending on the embedding. It is clear that nothing can change if one takes as embedding spacetime the direct product $\mathbb{R}^{1,n} \times \mathcal{M}$ with any space(-time) $\mathcal{M}$, extends functions to be constant on $\mathcal{M}$ and only deforms the first $n$ coordinates of the $\mathbb{R}^{1,n}$ factor as in our case. For a more general embedding, it is unclear what exactly might happen. Even if we keep the same embedding spacetime $\mathbb{R}^{1,n}$ but change the way functions are extended from de~Sitter to the full $\mathbb{R}^{1,n}$ or which coordinates are deformed, the results will change. However, since $X^\mu \to x^\mu$ as the Hubble parameter $H \to 0$ while $X^n$ diverges in this limit, any such change will destroy the commutative diagram of Fig.~\ref{pic:commuting}. That is, our choice is the only one where the limit $H \to 0$ and the deformation commute, apart from the previously mentioned phenomenological advantages. Another criterion that one could employ to reduce the ambiguity in the embedding and the deformation is dynamical locality~\cite{FVa,FVb}, which requires that descriptions of local physics based on kinematical and dynamical considerations should coincide. However, to use this criterion one would first need to verify whether the deformed theories that we study correspond to local quantum field theories.\footnote{We thank R.~Verch for this comment.} 

For simplicity, we had considered only a massless, conformally coupled scalar field, where the de~Sitter two-point function is given by a conformal rescaling of the Minkowski two-point function. However, already in this simplest example this relation only holds for the undeformed two-point function but not for the deformed one, since in flat space the deformed two-point function is equal to the undeformed one, whereas they differ in de~Sitter spacetime. For a scalar field with more general mass $m$ or curvature coupling $\xi$, it is well-known~\cite{AllenDeSitter} that for $m^2 + \xi R = m^2 + \xi n (n-1) H^2 > 0$ the undeformed two-point function only depends on the de~Sitter invariant $Z$, which for equal times only depends on the physical (co-moving) coordinates. Hence, at late times $t \to \infty$ we would expect results very similar to Corr.~\ref{corr:latetime_tpf}, namely that the corrections coming from the non-commutativity are exponentially suppressed with respect to the undeformed two-point function.

Lastly, it would be very interesting to connect our proposed deformation to the Weyl quantization for curved spacetimes defined in the recent Ref.~\onlinecite{DSL}. We first remark that while heuristically the product of two Weyl quantized functions $f$ and $g$ should be related to the Weyl quantization of the deformed product (or Rieffel product in the flat case) of the two functions, the relation cannot be totally straightforward since Weyl quantization concerns functions defined on phase space, i.e., depending on both coordinates and momenta, while the deformed product we presented only concerns functions defined on the manifold with no dependence on momenta. One possibility to make a comparison would be given by considering instead of just de~Sitter space its direct product with $\mathbb{R}^{1,n-1}$, which can be interpreted as momentum space. The embedding of the whole manifold into $\mathbb{R}^{2,2n-1}$ is then straightforward, and one could define a deformed product in total analogy to what we have presented before. However, the physical interpretation of such a construction is not clear, since as we have seen the ``momentum operators'' $P_\mu(X)$ that induce translations of the quantum field operators depend on position, and no global momentum exists. Since for quantum field theory in curved spacetimes the two-point function plays an essential role, a much more relevant method would be the comparison of the deformed two-point functions. For this, one would need to apply the Weyl quantization and the resulting star product of Ref.~\onlinecite{DSL} to free scalar quantum field theory, with the scalar field and its conjugate momentum in canonical quantization playing the role of position and momenta. The resulting two-point function could then be compared with the one of our non-commutative field $\phi_\Theta$, and one could verify whether there is a choice of $\Theta$ that makes the result agree (at least to first order) with the one that would result using the methods of Ref.~\onlinecite{DSL}.

\begin{acknowledgments}
This work has been funded by the Deutsche Forschungsgemeinschaft (DFG, German Research Foundation) --- project nos. 415803368 and 406116891 within the Research Training Group RTG 2522/1. We thank S. Franchino-Vi\~nas for comments and the anonymous Referee for various suggestions that helped improve the manuscript.
\end{acknowledgments}

\section*{Data availability}
In accordance with AIP Publishing policy, the authors declare that ``Data sharing is not applicable to this article as no new data were created or analyzed in this study.''

\appendix

\section{Oscillatory Integrals}
\label{oi}

The definitions and results of oscillatory integrals that we use are given in this section. First we define an oscillatory integral as in Ref.~\onlinecite[Ch.~(7.8)]{H}.
\begin{definition}
\label{oipf}
Let $X \subset \mathbb{R}^n$ be open and let $\Gamma$ be an open cone on $X \times \left( \mathbb{R}^N \setminus \{0\} \right)$ for some $N$, i.e., for all $(x,y) \in \Gamma$ also $(x,\lambda y) \in \Gamma$ for all $\lambda > 0$. We say that a function $\phi \in C^\infty(\Gamma)$ is a \textbf{phase function} in $\Gamma$ if
\begin{itemize}
\item $\phi(x, \lambda y) = \lambda \phi(x, y)$ for all $(x,y) \in \Gamma$ and $\lambda > 0$ \eqend{,}
\item $\Im \phi(x, y) \geq 0$ for all $(x,y) \in \Gamma$ \eqend{,}
\item $\nabla \phi(x, y) \neq 0$ for all $(x,y) \in \Gamma$ \eqend{.}
\end{itemize}
Then an integral of the form (see Ref.~\onlinecite[Eq.~7.8.1]{H})
\begin{equation}
\label{eq:iopf_integral}
\int \mathe^{\mathi \phi(x,y)} b(x,y) \total^N y
\end{equation}
is called an \textbf{oscillatory integral}.
\end{definition}
Another definition that we use is that of a symbol, see Ref.~\onlinecite[Def.~7.8.1]{H}.
\begin{definition}
Let $m$, $\rho$, $\delta$ be real numbers with $0 < \rho \leq 1$ and $0 \leq \delta < 1$. Then we denote by $S^{m}_{\rho,\delta}(X\times \mathbb{R}^N)$ the set of all $b\in C^\infty\left( X \times \mathbb{R}^N \right)$ such that for every compact set $K \subset X$ and all $\alpha$, $\beta$ there exists some constant $C_{\alpha,\beta,K}$ such that the estimate
\begin{equation}
\abs{ \partial^{\beta}_x \partial^{\alpha}_y b(x,y) } \leq C_{\alpha,\beta,K}(1+\abs{y})^{m - \rho \abs{\alpha} + \delta \abs{\beta}}
\end{equation}
is valid for all $(x,y) \in K \times \mathbb{R}^N$. The elements $S^{m}_{\rho,\delta}$ are called symbols of order $m$ and type $\rho,\delta$.
\end{definition}
It is proven in Ref.~\onlinecite[Thm.~7.8.2]{H} (see also Refs.~\onlinecite{LW,Jo}), that if $b \in S^{m}_{\rho,\delta}$ and $m < -N+1$ the oscillatory integral~\eqref{eq:iopf_integral} converges to a well-defined function. In the case $m \geq -N+1$, the oscillatory integral has to be defined in a distributional manner.

The Rieffel product~\eqref{eq:rfp}, given as an oscillatory integral with phase function $\phi(x,y) = x \cdot y$, thus either converges to a well-defined function or can still be defined in a distributional manner if the integrand is a symbol. Hence, we define the domain $\mathcal{D}$ to be the set of all functions $f$, $g$ such that
\begin{equation}
\tau_{\Theta x}(f) \tau_y(g) \in S^{m}_{\rho,\delta} \eqend{.}
\end{equation}

\section{Determination of the operators \texorpdfstring{$P_\mu$}{P\textmu}}
\label{sec:caop}

Here we explicitly determine the operators defined by equations~\eqref{eq:p0phi_comm} and~\eqref{eq:piphi_comm} in terms of the creation and annihilation operators. We make the ansatz
\begin{equation}
P_\mu(X) = \int \left[ p_\mu(\vec{q},X) a^\dagger(\vec{q}) a(\vec{q}) + r_\mu(\vec{q},X) a(\vec{q}) a(-\vec{q}) + r^*_\mu(\vec{q},X) a^\dagger(\vec{q}) a^\dagger(-\vec{q}) \right] \frac{\total^{n-1} q}{(2\pi)^{n-1}} \eqend{,}
\end{equation}
and compute
\begin{splitequation}
\left[ P_\mu(X), \phi(X) \right] &= \int \bigg[ \left[ p_\mu(\vec{p},X) f(\vec{p},t) - r_\mu(\vec{p},X) f^*(-\vec{p},t) - r_\mu(-\vec{p},X) f^*(-\vec{p},t) \right] a(\vec{p}) \\
&\qquad+ \left[ - p_\mu(-\vec{p},X) f^*(-\vec{p},t) + r^*_\mu(-\vec{p},X) f(\vec{p},t) + r^*_\mu(\vec{p},X) f(\vec{p},t) \right] a^\dagger(-\vec{p}) \bigg] \\
&\qquad\qquad\times \mathe^{\mathi \vec{p} \vec{x}} \frac{\total^{n-1} p}{(2\pi)^{n-1}} \eqend{.}
\end{splitequation}
Comparing with~\eqref{eq:p0phi_comm} and~\eqref{eq:piphi_comm}, we obtain
\begin{equations}[eq:app:caop]
p_0(\vec{p},X) &= - \frac{1}{H^2 X^n (X^0 + X^n)} \left[ \abs{\vec{p}} + H X^b \vec{p}_b \right] \eqend{,} \\
r_0(\vec{p},X) &= - \frac{n-2}{4} \frac{\mathi}{X^n} \exp\left[ 2 \mathi \frac{1 - H (X^0 + X^n)}{H (X^0 + X^n)} \frac{\abs{\vec{p}}}{H} \right] \eqend{,} \\
p_a(\vec{p},X) &= \frac{1}{H (X^0 + X^n)} \left[ \vec{p}_a + \frac{X_a}{H X^n (X^0 + X^n)} \left[ \abs{\vec{p}} + H X^b \vec{p}_b \right] \right] \eqend{,} \\
\begin{split}
r_a(\vec{p},X) &= \frac{n-2}{4} \frac{\mathi}{X^n (X^0 + X^n)} X_a \exp\left[ 2 \mathi \frac{1 - H (X^0 + X^n)}{H (X^0 + X^n)} \frac{\abs{\vec{p}}}{H} \right] \\
&= - \frac{X_a}{X^0 + X^n} r_0(\vec{p},X) = - H \vec{x}_a r_0(\vec{p},X) \eqend{,}
\end{split}
\end{equations}
and in the flat-space limit $H \to 0$ this has the correct limit
\begin{equation}
p_0(\vec{p},X) \to - \abs{\vec{p}} \eqend{,} \qquad r_0(\vec{p},X) \to 0 \eqend{,} \qquad p_a(\vec{p},X) \to \vec{p}_a \eqend{,} \qquad r_a(\vec{p},X) \to 0 \eqend{.}
\end{equation}

\section{Commutativity of the operators \texorpdfstring{$P_\mu$}{P\textmu}}
\label{app:pmu_comm}

We compute the commutator of two operators $P_\mu$ at possibly different points $X$ and $Y$.
\begin{splitequation}
\label{eq:app_pmu_comm_def}
\left[ P_\mu(X), P_\nu(Y) \right] &= \int \bigg[ r_{\mu\nu}(\vec{q},X,Y) a(\vec{q}) a(-\vec{q}) - r_{\mu\nu}^*(-\vec{q},X,Y) a^\dagger(\vec{q}) a^\dagger(-\vec{q}) \\
&\qquad+ p_{\mu\nu}(\vec{q},X,Y) a(-\vec{q}) a^\dagger(-\vec{q}) + p_{\mu\nu}(\vec{q},X,Y) a^\dagger(\vec{q}) a(\vec{q}) \Bigg] \frac{\total^{n-1} q}{(2\pi)^{n-1}} \eqend{,}
\end{splitequation}
with
\begin{equations}
p_{\mu\nu}(\vec{q},X,Y) &\equiv [ r_\mu(\vec{q},X) + r_\mu(-\vec{q},X) ] r^*_\nu(\vec{q},Y) - [ r^*_\mu(\vec{q},X) + r^*_\mu(-\vec{q},X) ] r_\nu(-\vec{q},Y) \eqend{,} \\
r_{\mu\nu}(\vec{q},X,Y) &\equiv [ r_\mu(\vec{q},X) + r_\mu(-\vec{q},X) ] p_\nu(-\vec{q},Y) - p_\mu(-\vec{q},X) [ r_\nu(\vec{q},Y) + r_\nu(-\vec{q},Y) ] \eqend{,}
\end{equations}
which does not obviously vanish for any choice of points $X$ and $Y$. However, we compute
\begin{equations}
p_{00}(\vec{p},X,Y) &= 4 \mathi \, \Im\left[ r_0(\vec{p},X) r^*_0(\vec{p},Y) \right] \eqend{,} \\
p_{0b}(\vec{p},X,Y) &= - H \vec{y}_b p_{00}(\vec{p},X,Y) \eqend{,} \\
p_{a0}(\vec{p},X,Y) &= - H \vec{x}_a p_{00}(\vec{p},X,Y) \eqend{,} \\
p_{ab}(\vec{p},X,Y) &= H^2 \vec{x}_a \vec{y}_b p_{00}(\vec{p},X,Y) \eqend{,}
\end{equations}
so the $p_{\mu\nu}$ vanish for $X = Y$, i.e., when the operators are taken at the same point $X$. We then obtain further
\begin{equations}
r_{00}(\vec{p},X,X) &= 0 \eqend{,} \\
r_{0b}(\vec{p},X,X) &= - \frac{2}{H (X^0 + X^n)} r_0(\vec{p},X) \vec{p}_b \eqend{,} \\
r_{a0}(\vec{p},X,X) &= \frac{2}{H (X^0 + X^n)} r_0(\vec{p},X) \vec{p}_a \eqend{,} \\
r_{ab}(\vec{p},X,X) &= \frac{2}{(X^0 + X^n)} r_0(\vec{p},X) \left( \vec{x}_a \vec{p}_b - \vec{x}_b \vec{p}_a \right) \eqend{.}
\end{equations}
Since $r_0(\vec{p},X)$ is an even function of $\vec{p}$, the $r_{\mu\nu}(\vec{p},X,X)$ are odd functions of $\vec{p}$, and because the annihilation operators commute among themselves (as well as the creation operators) the integrand in the commutator~\eqref{eq:app_pmu_comm_def} is an odd function such that the integral vanishes. That is, we have
\begin{equation}
\label{eq:app_pmu_comm}
\left[ P_\mu(X), P_\nu(X) \right] = 0 \eqend{.}
\end{equation}

\section{Proof of Theorem~\ref{thm:tpf}}
\label{proof:tpf}

\begin{proof}
To first order in the non-commutative parameter $\Theta$ we have~\eqref{eq:phitheta_firstorder}
\begin{splitequation}
\phi_\Theta(X) &= \phi(X) - \mathi \Theta^{\mu\nu} \left[ P_\mu(X), \phi(X) \right] P_\nu(X) + \bigo{\Theta^2} \\
&= \phi(X) - \mathi \Theta^{\mu\nu} \Big( \left[ P_\mu(X) P_\nu(X), \phi(X) \right] - P_\mu(X) \left[ P_\nu(X), \phi(X) \right] \Big) + \bigo{\Theta^2} \\
&= \phi(X) + \mathi \Theta^{\mu\nu} P_\mu(X) \left[ P_\nu(X), \phi(X) \right] + \bigo{\Theta^2} = {}_\Theta\phi(X) + \bigo{\Theta^2} \eqend{,}
\end{splitequation}
where we used that $\Theta^{\mu\nu} P_\mu(X) P_\nu(X) = \Theta^{\mu\nu} \left[ P_\mu(X), P_\nu(X) \right]/2 = 0$~\eqref{eq:app_pmu_comm}, which confirms the general equality $\phi_\Theta = {}_\Theta\phi$ to first order.

We thus can compute the two-point function up to first order in $\Theta$:
\begin{splitequation}
\label{eq:app:tpf_2pf}
&\bra{0} \phi_\Theta(X) \phi_\Theta(Y) \ket{0} = \bra{0} \phi(X) \phi(Y) \ket{0} + \Theta^{\mu\nu} \bra{0} \partial_{X^\mu} \phi(X) P_\nu(X) \phi(Y) \ket{0} \\
&\hspace{10em}+ \Theta^{\mu\nu} \bra{0} \phi(X) \partial_{Y^\mu} \phi(Y) P_\nu(Y) \ket{0} + \bigo{\Theta^2} \\
&\quad= \exp\left[ - \frac{n-2}{2} H (s+t) \right] \int \frac{1}{2 \abs{\vec{p}}} \exp\left[ \mathi \big[ \exp( - H t ) - \exp( - H s ) \big] \frac{\abs{\vec{p}}}{H} \right] \mathe^{\mathi \vec{p} (\vec{x} - \vec{y})} \frac{\total^{n-1} p}{(2\pi)^{n-1}} \\
&\qquad- \frac{n-2}{2} \Theta^{0b} \left[ \exp( - H t ) \frac{1}{X^n} + \exp( - H s ) \frac{1}{Y^n} \right] \exp\left[ - \frac{n-2}{2} H (s+t) \right] \\
&\qquad\qquad\times \int \frac{\vec{p}_b}{2 \abs{\vec{p}}} \exp\left[ \mathi \big[ \exp( - H t ) - \exp( - H s ) \big] \frac{\abs{\vec{p}}}{H} \right] \mathe^{\mathi \vec{p} (\vec{x} - \vec{y})} \frac{\total^{n-1} p}{(2\pi)^{n-1}} \\
&\qquad+ \frac{n-2}{2} H \Theta^{ab} \left[ \exp( - H t ) \frac{\vec{x}_a}{X^n} + \exp( - H s ) \frac{\vec{y}_a}{Y^n} \right] \exp\left[ - \frac{n-2}{2} H (s+t) \right] \\
&\qquad\qquad\times \int \frac{\vec{p}_b}{2 \abs{\vec{p}}} \exp\left[ \mathi \big[ \exp( - H t ) - \exp( - H s ) \big] \frac{\abs{\vec{p}}}{H} \right] \mathe^{\mathi \vec{p} (\vec{x} - \vec{y})} \frac{\total^{n-1} p}{(2\pi)^{n-1}} + \bigo{\Theta^2} \eqend{.}
\end{splitequation}

The $\vec{p}$ integral is easily done in spherical coordinates using that
\begin{equation}
\int f(\abs{\vec{p}}) \mathe^{\mathi \vec{p} \vec{x}} \frac{\total^{n-1} p}{(2\pi)^{n-1}} = \frac{1}{(2 \pi)^\frac{n-1}{2}} \int_0^\infty f(p) p^{n-2} \left( p \abs{\vec{x}} \right)^\frac{3-n}{2} \bessel{J}{\frac{n-3}{2}}{p \abs{\vec{x}}} \total p \eqend{,}
\end{equation}
where J is the Bessel function, and gives for the undeformed two-point function
\begin{splitequation}
\bra{0} \phi(X) \phi(Y) \ket{0} &= \frac{1}{2 (2 \pi)^\frac{n-1}{2} r^{n-2}} \exp\left[ - \frac{n-2}{2} H (s+t) \right] \\
&\qquad\times \lim_{\epsilon \to 0^+} \int_0^\infty \exp\left[ \mathi \frac{\exp( - H t ) - \exp( - H s )}{r H} p - \epsilon p \right] p^\frac{n-3}{2} \bessel{J}{\frac{n-3}{2}}{p} \total p \\
&= \frac{\Gamma\left( \frac{n-2}{2} \right)}{4 \pi^\frac{n}{2}} \lim_{\epsilon \to 0^+} \left[ r^2 \exp\left[ H (s+t) \right] - \frac{2}{H^2} \Big( \cosh\left[ H (s-t) \right] - 1 \Big) + \mathi \epsilon \sgn(t-s) \right]^\frac{2-n}{2} \\
&= \frac{\Gamma\left( \frac{n-2}{2} \right) H^{n-2}}{2 (2\pi)^\frac{n}{2}} \lim_{\epsilon \to 0^+} \left[ 1 - Z(X,Y) + \mathi \epsilon \sgn(t-s) \right]^\frac{2-n}{2} \eqend{.}
\end{splitequation}
with $r \equiv \abs{\vec{x}-\vec{y}}$. This is of course the well-known result for the two-point function~\cite{AllenDeSitter}, which here merely serves as a check on the computation. Performing the remaining spatial derivatives in equation~\eqref{eq:app:tpf_2pf} coming from the $\vec{p}_b$ in the spatial integrals, we obtain after some rearrangements the result~\eqref{eq:tpf}.
\end{proof}

\section*{References}

\bibliography{allliterature1}

\end{document}